\documentclass[10pt]{article}
\usepackage{amsmath}
\usepackage{amssymb}

\usepackage{amsfonts}
\usepackage[english]{babel}
\usepackage{amsthm}
\usepackage[all]{xy}

\newtheorem{theorem}{Theorem}
\newtheorem{lemma}{Lemma}
\newtheorem{definition}{Definition}

\newtheorem{remark}{Remark}

\newcommand{\mR}{\mathbb{R}}
\newcommand{\mC}{\mathbb{C}}
\newcommand{\mN}{\mathbb{N}}
\newcommand{\mE}{\mathbb{E}}

\newcommand{\mP}{\mathbb{P}}
\newcommand{\mQ}{\mathbb{Q}}

\newcommand{\mV}{\mathbb{V}}
\newcommand{\mW}{\mathbb{W}}

\newcommand{\cH}{\mathcal{H}}

\newcommand{\cP}{\mathcal{P}}

\newcommand{\ux}{\underline{x}}
\newcommand{\uxb}{\underline{x \grave{}}}

\newcommand{\pj}{\partial_{x_j}}
\newcommand{\pjb}{\partial_{{x \grave{}}_{j}}}

\newcommand{\gso}{\mathfrak{so}}
\newcommand{\gsp}{\mathfrak{sp}}

\title{Spherical harmonics and integration in superspace II}
\author{H.\ De Bie\thanks{Corresponding author} \thanks{Ph.D. Fellow of the Research Foundation - Flanders (FWO), E-mail: {\tt Hendrik.DeBie@UGent.be}} \and  D.\ Eelbode\thanks{E-mail: {\tt 	David.Eelbode@ua.ac.be}} \and F.\ Sommen\thanks{E-mail: {\tt fs@cage.ugent.be}}
}

\date{\small{H. De Bie and F. Sommen: Clifford Research Group -- Department of Mathematical Analysis}\\
\small{Faculty of Engineering -- Ghent University\\ Krijgslaan 281, 9000 Gent,
Belgium}\\ \vspace{5mm}
\small{D. Eelbode: Department of Mathematics and Computer Science\\
University of Antwerp\\
Middelheimlaan 1, 2020 Antwerpen, Belgium}
}

\begin{document}
\maketitle
\begin{abstract}
The study of spherical harmonics in superspace, introduced in [J. Phys. A: Math. Theor. 40 (2007) 7193-7212], is further elaborated. A detailed description of spherical harmonics of degree $k$ is given in terms of bosonic and fermionic pieces, which also determines the irreducible pieces under the action of $SO(m) \times Sp(2n)$. In the second part of the paper, this decomposition is used to describe all possible integrations over the supersphere. It is then shown that only one possibility yields the orthogonality of spherical harmonics of different degree. This is the so-called Pizzetti-integral of which it was shown in [J. Phys. A: Math. Theor. 40 (2007) 7193-7212] that it leads to the Berezin integral.
\end{abstract}
\\
\noindent
\textbf{Mathematics Subject ClassiÞcation :} 30G35, 58C50

\newpage

\section{Introduction}

This paper deals with the study of superspaces, from a function theoretical point of view. Superspaces, which play an important role in modern theoretical physics, can be considered as co-ordinate spaces for theories exhibiting supersymmetry: they are spaces in which both commuting and anti-commuting variables are considered, corresponding to respectively bosonic and fermionic degrees of freedom.\\
\noindent 
Mathematically speaking, there are several different approaches to this theory: the two most important ones are based on methods from algebraic geometry, see e.g. \cite{MR0208930,MR732126,MR0580292,MR565567}, or methods from differential geometry, see e.g. \cite{MR778559,MR574696}. Recently however, we have started investigating superspaces from yet another point of view, focussing on close connections with the framework of classical harmonic analysis and Clifford analysis. The latter is a higher-dimensional function theory in which spin-invariant differential operators, such as the Dirac operator or the Rarita-Schwinger operator, are studied (see e.g. \cite{MR697564,MR1169463, BSSVL}). This connection essentially reveals itself in the existence of certain differential operators, acting between functions defined on superspaces and establishing a realization of either the simple Lie algebra $\mathfrak{sl}_2$, or the Lie superalgebra $\mathfrak{osp}(1 \vert 2)$. The basic framework underlying these observations was obtained in \cite{DBS1,DBS4,DBS2}, and in \cite{DBS5} we applied these results to construct an integral over the supersphere. This integral, which was inspired by an old result \cite{PIZZETTI} of Pizzetti, allowed for the generalization to superspace of several classical interesting results such as orthogonality of spherical harmonics of different degree, Green's theorem, the Funk-Hecke theorem, etc. Combining this integral over the supersphere with a generalization of integration in spherical co-ordinates, we were not only able to construct an integral over the whole superspace but also to prove the equivalence with the Berezin integral (see \cite{MR732126}). \\
\noindent
In the present paper, we will show that the construction of this integral can also be done using the crucial fact that the differential operators leading to a realization of $\mathfrak{sl}(2)$ in superspace are invariant with respect to a suitable action of the group $SO(m) \times Sp(2n)$. This can then not only be seen as an alternative way to justify the definitions and constructions obtained in \cite{DBS5}, which were based on a formal analogy with the classical case, but also as a means to give a precise mathematical meaning to the Berezin integral, based on the symmetry underlying our theoretical framework. In order to achieve this, we will investigate how spaces of spherical harmonics in superspace behave under the regular representation of the aforementioned group, and decompose these spaces into irreducible subspaces. Using the fact that an integral over the supersphere (which is invariant under the regular group action and satisfies a natural constraint inspired by the algebraic equation for the supersphere) can be seen as a linear functional, we will then be able to invoke Schur's Lemma in order to obtain the most general definition of an integral over the supersphere. Surprisingly enough, this will lead to several possibilities spanning a finite-dimensional vectorspace, but the natural constraint that spherical harmonics of different degree be orthogonal will lead to a unique definition. This will then precisely be the analogue of Pizzetti's formula, i.e. the Berezin integral in superspace. Note that other attempts to relate the Berezin integral to more familiar types of integration have been made in e.g. \cite{MR784620} and \cite{MR825156}.\\
\noindent
The paper is organized as follows. In section 2 we recall the basic theory of harmonic analysis in superspace necessary for this paper. In section 3 we discuss the action of $SO(m) \times Sp(2n)$, followed in section 4 by a discussion of the purely bosonic and purely fermionic case. Next, in section 5 we construct projection operators selecting harmonic components out of spaces of homogeneous polynomials.  In section 6 we give the general decomposition of spaces of spherical harmonics under the group action. In section 7 we define integration over the supersphere and show how this definition already restricts possible integrations to a finite dimensional vectorspace. In section 8 we give some examples of integrals over the supersphere. In the last section we discuss how we can distinguish between these different types of integration, thus obtaining a unique determination of the Berezin integral.

\section{Harmonic analysis in superspace}

We consider the algebra $\cP$, generated by $m$ commuting variables $x_i$ and $2n$ anti-commuting variables ${x\grave{}}_i$, subject to the following relations
\[
\begin{array}{l}
x_i x_j =  x_j x_i\\
{x \grave{}}_i {x \grave{}}_j =  - {x \grave{}}_j {x \grave{}}_i\\
x_i {x \grave{}}_j =  {x \grave{}}_j x_i.\\
\end{array}
\]
In other words, we have that
\[
\cP = \mR[x_1, \ldots , x_m] \otimes \Lambda_{2n}
\]
with $\Lambda_{2n}$ the Grassmann algebra generated by the ${x \grave{}}_i$.
We will equip this algebra with the necessary operators to obtain a representation of harmonic analysis in superspace.

First of all, we introduce a super Laplace operator by putting
\[
\Delta  = 4 \sum_{j=1}^n \partial_{{x \grave{}}_{2j-1}} \partial_{{x \grave{}}_{2j}}-\sum_{j=1}^{m} \pj^2
\]
and similarly a bosonic, resp. fermionic Laplace operator $\Delta_b$, $\Delta_f$ by
\[
\begin{array}{lll}
\Delta_b&=&-\sum_{j=1}^{m} \pj^2 \\
\vspace{-2mm}\\
\Delta_f&=& 4 \sum_{j=1}^n \partial_{{x \grave{}}_{2j-1}} \partial_{{x \grave{}}_{2j}}.
\end{array}
\]

We also have the following generalization of the length squared of a vector in Euclidean space
\[
x^2 = \sum_{j=1}^n {x\grave{}}_{2j-1} {x\grave{}}_{2j}  -  \sum_{j=1}^m x_j^2
\]
with its respective bosonic and fermionic components $\ux^2,\uxb^2$ defined by
\[
\begin{array}{lll}
\ux^2&=&-\sum_{j=1}^m x_j^2 \\
\vspace{-2mm}\\
\uxb^2&=& \sum_{j=1}^n {x\grave{}}_{2j-1} {x\grave{}}_{2j}.
\end{array}
\]

Finally we introduce the super-Euler operator by
\[
\mE=\sum_{j=1}^m x_j \pj +\sum_{j=1}^{2n} {x \grave{}}_j \pjb
\]
allowing us to put a grading on $\cP$, viz.
\[
\cP = \bigoplus_{k=0}^{\infty} \cP_k, \quad \cP_k=\left\{ \omega \in \cP \; | \; \mE \omega=k \omega \right\}.
\]

Now let us calculate the following:
\[
\Delta(x^2) = 2 (m-2n) = 2M
\]
where $M$ is the so-called super-dimension. This parameter will play an important role in the sequel. It also has a nice physical interpretation, see \cite{DBS3}.

Putting $X = x^2/2$, $Y=-\Delta/2$ and $H = \mE + M/2$, we can calculate the following commutators:
\begin{eqnarray}
[H,X] &=& 2X\\
\left[H,Y\right] &=& -2Y\\
\left[X,Y\right] &=& H
\label{lierelations}
\end{eqnarray}
proving that $X$, $Y$ and $H$ span the Lie algebra $\mathfrak{sl}_2$ and that we have indeed a representation of harmonic analysis in superspace (see e.g. \cite{MR1151617}). The present framework can also be extended by introducing the super Dirac operator and the super vector variable, thus leading to a representation of Clifford analysis in superspace (see \cite{DBS4}). In that case the basic operators span the Lie superalgebra $\mathfrak{osp}(1|2)$, of which the even part is $\mathfrak{sl}_2$.

Now we have the following basic lemma.
 
\begin{lemma}
For $R_k \in \cP_k$ one has the following relation:
\begin{equation}
\Delta (x^{2t} R_{k})= 2t(2k+M+2t-2) x^{2t-2} R_k + x^{2t} \Delta R_k.
\label{basicformula}
\end{equation}
\label{iteratedlaplace}
\end{lemma}

\begin{proof}
Using the formula $[ \Delta,  x^2 ] =  4 \mE + 2 M$ (see (\ref{lierelations})) we obtain
\begin{equation*}
\Delta(x^2 R_k) = (4k+2M) R_k + x^2 \Delta R_k,
\end{equation*}
since $R_k$ is a homogeneous polynomial of degree $k$. The desired result follows by induction.
\end{proof}

As a consequence of this lemma, we have the following
\begin{lemma}
If $R_{2t} \in \cP_{2t}$ then the following holds:
\begin{equation}
\Delta^{t+1}(x^2 R_{2t}) = 4(t+1)(M/2+t) \Delta^{t}( R_{2t}).
\label{relationslaplace}
\end{equation}
\end{lemma}

\begin{proof}
See \cite{DBS5}.
\end{proof}

Now we define spherical harmonics in superspace as homogeneous null-solutions of the super Laplace operator.
\begin{definition}
A (super)-spherical harmonic of degree $k$ is an element $H_k \in \cP$ satisfying
\begin{eqnarray*}
\Delta H_k&=&0\\
\mE H_k &=& k H_k, \quad \mbox{i.e. } H_k \in \cP_k.
\end{eqnarray*}
The space of spherical harmonics of degree $k$ will be denoted by $\cH_k$.
\end{definition}

In the purely bosonic case (i.e. the case with only commuting variables) we denote $\cH_k$ by $\cH_k^b$, in the purely fermionic case (the case with only anti-commuting variables) by $\cH_k^f$. If necessary for clarity, the variables under consideration will be mentioned also, e.g. $\cH_k^f({x\grave{}}_{1}, \ldots, {x\grave{}}_{2n})$.

Iteration of lemma \ref{iteratedlaplace} immediately leads to the following result.

\begin{lemma}
Let $H_k \in \cH_k$ and $M \not \in -2 \mN$. Then for all $i,j,k \in \mN$ one has that
\[
\Delta^i(x^{2j} H_k) = \left\{
\begin{array}{l}
c_{i,j,k} x^{2j-2i} H_k, \quad i \leq j\\
\vspace{-2mm}\\
0, \quad i > j
\end{array}
\right.
\]
with
\[
c_{i,j,k} = 4^{i} \frac{j!}{(j-i)!} \frac{\Gamma(k+M/2+j)}{\Gamma(k+M/2+j-i)}.
\]
\label{laplonpieces}
\end{lemma}

Using this lemma we can prove the following Fischer decomposition of super-polynomials (for the classical case see e.g. \cite{MR0229863}):

\begin{theorem}[Fischer decomposition]
Suppose $M \not \in -2 \mN$. Then $\cP_k$ decomposes as
\begin{equation}
\cP_k = \bigoplus_{i=0}^{\left\lfloor \frac{k}{2} \right\rfloor} x^{2i} \cH_{k-2i}.
\end{equation}
\label{fdecomp}
\end{theorem}

\begin{proof}
See \cite{DBS5}.
\end{proof}
Note that this theorem fails whenever $M\in -2\mN$ and $m \neq 0$, because then the spaces $x^{2i} \cH_{k-2i}$ are no longer disjoint (we end up in the poles of the Gamma function appearing in lemma \ref{laplonpieces}). However, in the purely fermionic case we do have a Fischer decomposition, see section \ref{fermharmonics}.

Finally let us calculate the dimensions of the spaces of spherical harmonics in superspace. 

\begin{lemma}
One has that
\[
\dim \cH_k = \dim \cP_k - \dim \cP_{k-2}
\]
where
\[
\dim  \cP_k = \sum_{i=0}^{min(k,2n)} \binom{2n}{i} \binom{k-i+m-1}{m-1}
\]
and by definition $\dim  \cP_{-1} = \dim  \cP_{-2}=0$.
\end{lemma}

\begin{proof}
See \cite{DBS5}.
\end{proof}

In the purely bosonic and the purely fermionic case the previous lemma reduces to:
\[
\dim \cH_k^b = \binom{k+m-1}{m-1} - \binom{k-2+m-1}{m-1}
\]
and
\begin{eqnarray*}
\dim \cH_k^f &=& \binom{2n}{k} - \binom{2n}{k-2}, \quad k \leq n\\
\dim \cH_k^f &=& 0, \quad k > n.
\end{eqnarray*}

\begin{remark}
It is possible to refine the theory of spherical harmonics to spherical monogenics, i.e. null-solutions of the super Dirac operator, see \cite{DBS2}.
\end{remark}

\section{The group action}

In this section we show that the group $G = SO(m)\times Sp(2n)$ gives the appropriate action on $\cP$ for our purposes.

First of all, let us state the properties we want this group to exhibit:
\begin{itemize}
\item $G \cdot \cP_k \subseteq \cP_k$, i.e. the degree of homogeneity is preserved under $G$
\item $G \cdot x^2 = x^2$, i.e. the quadratic polynomial $x^2$ is invariant under $G$.
\end{itemize}

The first property restricts possible transformations to
\[
\begin{array}{lll}
y_i &=& \sum_{k=0}^{m} a_k^i x_k + \sum_{l=0}^{2n} b_l^i {x \grave{}}_l\\
\vspace{-2mm}\\
{y \grave{}}_j &=& \sum_{k=0}^{m} c_k^j x_k + \sum_{l=0}^{2n} d_l^j {x \grave{}}_l
\end{array}
\]
with $a_k^i,b_l^i,c_k^j,d_l^j \in \mR$. In matrix notation we have that
\[
y = S x = \left( \begin{array}{c|c} A&B\\ \hline \vspace{-3.5mm} \\C&D
\end{array}
 \right) x
\]
with $x =(x_1 \; \ldots \; x_m \;| \; {x \grave{}}_1 \; \ldots \; {x \grave{}}_{2n})^T$.

Similarly
\[
x^2 = x^T Q x = x^T \left( \begin{array}{c|c} -1&0\\ \hline \vspace{-3.5mm}  \\0&J
\end{array}
 \right)x
\]
with
\[
J = \left( \begin{array}{cccccc} 0&1/2&&&\\-1/2&0&&&\\&&\ddots&&\\&&&0&1/2\\&&&-1/2&0 
\end{array}
 \right)
\]

As we want that $y^2 = x^2$, this means that
\[
S^T Q S = Q.
\]

In terms of $A,B,C$ and $D$ this yields
\begin{eqnarray}
\label{eq1}
-A^T A + C^T J C &=& -1\\
\label{eq2}
-A^T B + C^T J D &=& 0\\
\label{eq3}
-B^T A + D^T J C &=& 0\\
\label{eq4}
-B^T B + D^T J D &=& J.
\end{eqnarray}

Adding to equation (\ref{eq1}) its transpose, taking into account that $J^T = -J$, leads to $A^T A =1$. Adding the transpose of (\ref{eq2}) to (\ref{eq3}) leads to $B^T A=0$ and thus to $B =0$ as $A$ is invertible. As $D^T J D = J$ we have that $D \in Sp(2n)$ and $D$ is invertible. Then (\ref{eq3}) becomes $D^T J C = 0$ and thus $C=0$ as also $J$ is invertible. We conclude that the group $G$ can be taken to be $SO(m) \times Sp(2n)$. The Lie algebra for $G$ is then given by the semi-simple algebra $\gso(m)\oplus\gsp_{\mC}(2n)$, the irreducible finite-dimensional representations of which are defined as tensors products $\mV_{\lambda} \otimes \mW_{\mu}$, where $\mV_{\lambda}$ denotes the irreducible $\gso(m)$-module with highest weight $\underline{\lambda}$ and $\mW_{\mu}$ the irreducible $\gsp_{\mC}(2n)$-module with highest weight $\underline{\mu}$.

It is also easily seen that $G \cdot \cH_k \subseteq \cH_k$ since the super Laplace operator is invariant under the action of $G$. However, the spaces $\cH_k$ are not irreducible under the action of $G$. As opposed to the purely bosonic and fermionic case, spaces of homogeneous (polynomial) solutions for the super Laplace operator do not lead to irreducible modules for the respective Lie algebra underlying the symmetry of the system. The complete decomposition will be presented in section \ref{decomp}.

\section{Representations of $SO(m)$ and $Sp(2n)$}

\subsection{Spherical harmonics in $\mR^m$}

The classical theory of spherical harmonics in $\mR^m$ is very well known (see e.g. \cite{MR0229863}). In this case the Fischer decomposition takes the following form:
\[
\cP_k = \bigoplus_{i=0}^{\left\lfloor \frac{k}{2} \right\rfloor} \ux^{2i} \cH_{k-2i}^b,
\]
where each space $\cH_k^b$ provides a model for the irreducible $\gso(m)$-module with highest weight vector $(k,0,\cdots,0)$.

Moreover, if $m > 2$ this is also the decomposition of the space of homogeneous polynomials of degree $k$ into irreducible pieces under the action of $SO(m)$. In the case where $m=2$, the spaces $\cH_{k-2i}^b$ all have dimension two and are also irreducible, when working over $\mR$ as is the case in this paper. When working over $\mC$, the spaces are still reducible because $SO(2)$ is abelian and all complex representations of abelian groups are one dimensional.

\subsection{Fermionic or symplectic harmonics}
\label{fermharmonics}

In the purely fermionic case there is also a Fischer decomposition (not included in theorem \ref{fdecomp}). It then takes the following form:
\begin{eqnarray}
\label{fermionicfischer}
\cP_k &=& \bigoplus_{i=0}^{\left\lfloor \frac{k}{2} \right\rfloor} \uxb^{2i} \cH_{k-2i}^f, \quad k\leq n\\
\label{fermfischer2}
\cP_{2n-k} &=& \bigoplus_{i=0}^{\left\lfloor \frac{k}{2} \right\rfloor} \uxb^{2n-2k+2i} \cH_{k-2i}^f, \quad k\leq n.
\end{eqnarray}
Each space $\cH_k^f$ provides a model for the fundamental representation for $\gsp_{\mC}(2n)$ with highest weight $(1,\cdots,1,0,\cdots,0)$, where the integer $1$ is to be repeated $k$ times ($k \leq n$) (see \cite{MR1153249}). This means that we have a decomposition of $\cP_k$ into irreducible pieces under the action of $Sp(2n)$.

Formula (\ref{fermionicfischer}) is proven in the same way as the general Fischer decomposition given in theorem \ref{fdecomp}. Formula (\ref{fermfischer2}) follows immediately by noticing that multiplication with $\uxb^{2n-2k}$ gives an isomorphism between $\cP_k$ and $\cP_{2n-k}$. 
From this fermionic Fischer decomposition we also obtain the following formulae:
\begin{eqnarray*}
\uxb^{2i} \cH_{k}^f &=& 0 \quad \mbox{for all $i > n-k$}\\
\uxb^{2i} \cH_{k}^f &\neq& 0 \quad \mbox{for all $i \leq n-k$}.
\end{eqnarray*}

\textbf{Construction of a basis:}

It is possible to construct a basis of $\cH_k^f$ by decomposing this space under the action of the subgroup $Sp(2) \times Sp(2n-2)$ of $Sp(2n)$. This leads to the following theorem.
\begin{theorem}
If $1<k\leq n$, then the space $\cH_k^f({x \grave{}}_1 , \ldots, {x \grave{}}_{2n})$ decomposes as
\begin{eqnarray*}
\cH_k^f({x \grave{}}_1 , \ldots, {x \grave{}}_{2n}) &=& \cH_k^f({x \grave{}}_3 , \ldots, {x \grave{}}_{2n}) \;  \oplus \; \cH_1^f({x \grave{}}_1 ,{x \grave{}}_2)\otimes \cH_{k-1}^f({x \grave{}}_3 , \ldots, {x \grave{}}_{2n})\\
&& \oplus \;  \left[{x \grave{}}_1 {x \grave{}}_2
 + \frac{1}{k-n-1} ({x \grave{}}_3 {x \grave{}}_4 + \ldots +{x \grave{}}_{2n-1} {x \grave{}}_{2n}) \right] \cH_{k-2}^f({x \grave{}}_3 , \ldots, {x \grave{}}_{2n}).
\end{eqnarray*}

If $k=1$, $\cH_1^f({x \grave{}}_1 , \ldots, {x \grave{}}_{2n})$ decomposes as
\[
\cH_1^f({x \grave{}}_1 , \ldots, {x \grave{}}_{2n}) = \cH_1^f({x \grave{}}_3 , \ldots, {x \grave{}}_{2n}) \oplus \cH_1^f({x \grave{}}_1 , {x \grave{}}_{2}).
\]
\end{theorem}

\begin{proof}
One easily checks, using formula (\ref{basicformula}), that each summand in the right-hand side is contained in $\cH_k^f({x \grave{}}_1 , \ldots, {x \grave{}}_{2n})$. Moreover, all three summands are mutually disjoint. The proof is completed if we check that both sides have the same dimension (as vectorspaces). Indeed, the dimension of the right-hand side is:
\begin{eqnarray*}
\dim RH &=& \dim \cH_k^f({x \grave{}}_3 , \ldots, {x \grave{}}_{2n}) + \dim \cH_1^f({x \grave{}}_1 ,{x \grave{}}_2) \dim \cH_{k-1}^f({x \grave{}}_3 , \ldots, {x \grave{}}_{2n}) + \dim \cH_{k-2}^f({x \grave{}}_3 , \ldots, {x \grave{}}_{2n})\\
&=& \binom{2n-2}{k} - \binom{2n-2}{k-2} + 2  \left( \binom{2n-2}{k-1} -  \binom{2n-2}{k-3} \right)\\
&& +  \binom{2n-2}{k-2} -  \binom{2n-2}{k-4}\\
&=& \binom{2n}{k} - \binom{2n}{k-2}
\end{eqnarray*}
after several applications of Pascal's rule. This equals the dimension of the left-hand side. The proof of the second statement is trivial.
\end{proof}

This theorem can be used to construct in an iterative way bases for $\cH_k^f({x \grave{}}_1 , \ldots, {x \grave{}}_{2n})$, since $\cH_0^f \cong \cP_0$ and $\cH_1^f \cong \cP_1$.

\section{Projection operators}
\label{projoperators}

We can explicitly determine the Fischer decomposition (see theorem \ref{fdecomp}) of a given polynomial $R_k \in \cP_k$.
To that end we have to construct a set of operators $\mP_i^k, i=0, \ldots, \left\lfloor \frac{k}{2} \right\rfloor$, that show the following behaviour:
\[
\mP_i^k (x^{2j} \cH_{k-2j}) = \delta_{ij} \cH_{k-2j},
\]
when acting on $\cP_{k}$.

There are several ways to obtain these operators. 
One way is to exploit the $\mathfrak{sl}_{2}$ relations in a similar way as in e.g. the Dunkl case (see \cite{MR1827871,MR2207700}). Then, after lengthy and quite technical computations one obtains

\begin{theorem}
The operator $\mP_i^k$ defined by 
\begin{equation}
\mP_i^k = \sum_{l=0}^{\lfloor k/2 \rfloor -i} \frac{(-1)^l}{4^{l+i} l! i!} (k-2i+M/2-1) \frac{\Gamma(k-2i-l-1+M/2)}{\Gamma(k-i+M/2)} x^{ 2l}\Delta^{i+l}
\label{exprai}
\end{equation}
satisfies
\[
\mP_i^k (x^{2j} \cH_{k-2j}) = \delta_{ij} \cH_{k-2j}.
\]
\label{projectionsfischer}
\end{theorem}

Note that this theorem is not valid if $M \in -2\mN$. In this case there is no Fischer decomposition and we end up in the poles of the Gamma function appearing in formula (\ref{exprai}).

Another way to obtain these projection operators is by using the Laplace-Beltrami operator in superspace. This operator is defined as
\[
\Delta_{LB} = x^2 \Delta - \mE(M-2 + \mE)
\]
and it is easy to check that it commutes with $x^{2}$. We then have
\begin{eqnarray*}
\Delta_{LB} x^{2j} \cH_k &=& x^{2j} \Delta_{LB} \cH_k\\
&=& -k(M-2+k) x^{2j} \cH_k,
\end{eqnarray*}
so each summand in the Fischer decomposition is an eigenspace of the Laplace-Beltrami operator. It is easy to see that  if $M > 0$, then for each $k \geq 0$ the eigenvalue $-k(M-2+k)$ is different. If $M < 0$ and $M$ is odd, then some values of $k$ give rise to the same eigenvalue. However, it is easily seen that this can only happen if they differ by an odd integer. In the purely fermionic case ($m=0$), then again $-k(M-2+k)$ is different for each $k$ because $0 \leq k \leq n$. This means that in the Fischer decomposition of $\cP_k$
\[
\cP_k = \bigoplus_{j=0}^{\left\lfloor \frac{k}{2} \right\rfloor} x^{2j} \cH_{k-2j}
\]
each summand has a different eigenvalue with respect to the Laplace-Beltrami operator. Hence the operator $\widetilde{\mP}_i^k$ defined by
\[
\widetilde{\mP}_i^k = \prod_{l=0, \;  l \neq i}^{\left\lfloor \frac{k}{2} \right\rfloor} \dfrac{\Delta_{LB} + (k-2l)(M-2+k-2l)}{2(i-l)(2k-2i-2l+M-2)}
\]
satisfies
\[
\widetilde{\mP}_i^k (x^{2j} \cH_{k-2j}) = \delta_{ij} x^{2j} \cH_{k-2j}.
\]
We also have that $\widetilde{\mP}_i^k = x^{2i} \mP^k_i$ when acting on $\cP_k$.

\section{Decomposition under the action of $SO(m) \times Sp(2n)$}
\label{decomp}

In this section the space $\cH_k$ will be decomposed into irreducible pieces under the action of the group $G=SO(m) \times Sp(2n)$. In view of the fact that irreducible representations for $G$, realized within the space $\cP_k$ of homogeneous polynomials in bosonic and fermionic variables, are tensor products of spaces of spherical and symplectic harmonics, it is natural to look for a subspace of $\cH_k$ inside the direct sum of subspaces of the form $\ux^{2i}\cH^b_p \otimes \uxb^{2j}\cH^f_q$. This is the subject of the following lemma.

\begin{lemma}
If $q < n$ and $k < n-q+1$, there exists a unique homogeneous polynomial $f_{k,p,q}=f_{k,p,q}(\ux^2,\uxb^2)$ of total degree $k$ such that $f_{k,p,q} \cH_p^b \otimes \cH_q^f \neq 0$ and
\[
\Delta (f_{k,p,q} \cH_p^b \otimes \cH_q^f) = 0,
\]
where the coefficient of $\ux^{2k}$ in $f_{k,p,q}$ is given by
\[
\frac{(n-q)!}{\Gamma(\frac{m}{2}+p+k)}.
\]
\label{polythm}
\end{lemma}

\begin{remark}
The restriction $q < n$ is necessary because in case $q = n$, all integer powers $\uxb^{2j}$ will act trivially on the space $\cH^f_n$, for $j > 0$. In the same vein, for $q < n$ there can only be a non-trivial action of $\uxb^{2j}$ on the space $\cH^f_q$ as long as $j \leq k$ with $k + q \leq n$. This explains the restricion $k < n-q+1$. 
\end{remark}

\begin{proof}
We first treat the case $p=q=0$. So we look for an $f_{k,0,0}(\ux^2,\uxb^2)$ of the following form:
\[
f_{k,0,0} = \sum_{i=0}^k a_i \ux^{2k-2i} \uxb^{2i}. 
\]

We now demand that 
\[
\Delta( f_{k,0,0} \cH_0^b \otimes \cH_0^f) = \Delta( f_{k,0,0})= 0.
\]
As $\Delta = \Delta_b + \Delta_f$ we find, using lemma \ref{iteratedlaplace}, that
\begin{eqnarray*}
\Delta( f_{k,0,0}) &=& \sum_{i=0}^{k-1} a_i (2k-2i)(m+2k-2i-2) \ux^{2k-2i-2} \uxb^{2i} + \sum_{i=1}^k a_i 2i (2i-2-2n) \ux^{2k-2i} \uxb^{2i-2}\\
&=&\sum_{i=0}^{k-1} \left[ a_i (2k-2i)(m+2k-2i-2) + a_{i+1} (2i+2) (2i-2n) \right]   \ux^{2k-2i-2} \uxb^{2i}.
\end{eqnarray*}

Hence we obtain the following recursion relation for the $a_i$:
\[
a_{i+1} = \frac{(2k-2i)(m+2k-2i-2)}{(2n-2i)(2i+2)} a_i
\]
which leads to the following explicit formula
\[
a_i = \frac{\Gamma(\frac{m}{2} + k)}{n!} \binom{k}{i} \frac{(n-i)!}{\Gamma(\frac{m}{2} + k-i)} a_0.
\]

If we now put $a_0 =  \frac{n!}{\Gamma(\frac{m}{2} + k)}$ we finally obtain
\begin{equation}
f_{k,0,0} =  \sum_{i=0}^k \binom{k}{i} \frac{(n-i)!}{\Gamma(\frac{m}{2} + k-i)} \ux^{2k-2i} \uxb^{2i}. 
\label{specpoly}
\end{equation}

The general case is now easily obtained. Indeed, the polynomial $f_{k,p,q}$ satisfying
\[
\Delta( f_{k,p,q} \cH_p^b \otimes \cH_q^f) = 0
\]
is found by the following substitutions in formula (\ref{specpoly}): $m \rightarrow m+2p$, $n \rightarrow n-q$ (see lemma \ref{iteratedlaplace}). This yields:
\[
f_{k,p,q} = \sum_{i=0}^k \binom{k}{i} \frac{(n-q-i)!}{\Gamma(\frac{m}{2} + p+ k-i)} \ux^{2k-2i} \uxb^{2i}.
\]

This final formula again explains the restrictions put on $k$ and $q$.
\end{proof}

We list some special cases, viz. the polynomials $f_{i,0,0}$ for $i = 1,2,3$:
\begin{eqnarray*}
f_{1,0,0} &=& \frac{n!}{\Gamma(\frac{m}{2}+1)} \left(\ux^2 + \frac{m}{2n} \uxb^2 \right)\\
f_{2,0,0} &=&\frac{n!}{\Gamma(\frac{m}{2}+2)} \left(\ux^4 + \frac{m+2}{n} \ux^2\uxb^2+ \frac{m (m+2)}{4n(n-1)}\uxb^4\right)\\
f_{3,0,0} &=&\frac{n!}{\Gamma(\frac{m}{2}+3)} \left(\ux^6 + \frac{3(m+4)}{2n} \ux^4\uxb^2+ \frac{ 3(m+2)(m+4)}{4n(n-1)}\ux^2\uxb^4 + \frac{m(m+2)(m+4)}{8n(n-1)(n-2)} \uxb^6 \right).
\end{eqnarray*}

Using some elementary identities for special functions, the polynomials obtained in lemma \ref{polythm} can be rewritten in terms of well-known orthogonal polynomials. One finds that the polynomial $f_{k,p,q}$ can be represented as
\[
f_{k,p,q} = -\pi\:\frac{(n-q)!}{\left(p+k+\frac{m}{2}\right)!}\:\frac{k!}{(q-n)_k}\:\ux^{2k}P^{(q-k-1,1-k-p-q-\frac{m}{2})}_k\left(1+2\frac{\uxb^2}{\ux^2}\right),
\]
with $P^{\alpha,\beta}_n(t)$ the Jacobi polynomial.

We now check that the following dimension formula holds:
\begin{lemma}
One has
\[
\dim \cH_{k} = \sum_{i=0}^{\min(n,k)} \dim{\cH^b_{k-i}} \dim{\cH^f_{i}} + \sum_{j=0}^{\min(n, k-1)-1} \sum_{l=1}^{\min(n-j,\lfloor \frac{k-j}{2} \rfloor)} \dim{\cH^b_{k-2l-j}} \dim{\cH^f_{j}}.
\]
\label{dimcheck}
\end{lemma}

\begin{proof}
The cases $k=1$ and $k=2$ are easily verified by explicitly writing the decomposition and plugging in the binomial factors. We then proceed by induction. We restrict ourselves to the case $ k \leq n$. We then need to prove that
\begin{equation}
\dim \cH_{k} = \sum_{i=0}^{k} \dim{\cH^b_{k-i}} \dim{\cH^f_{i}} + \sum_{j=0}^{k-2} \sum_{l=1}^{\lfloor \frac{k-j}{2} \rfloor} \dim{\cH^b_{k-2l-j}} \dim{\cH^f_{j}}.
\label{proofdimform}
\end{equation}

Now suppose the lemma holds for $\cH_{k-2}$, i.e.
\[
\dim \cH_{k-2} = \sum_{i=0}^{k-2} \dim{\cH^b_{k-i-2}} \dim{\cH^f_{i}} + \sum_{j=0}^{k-4} \sum_{l=1}^{\lfloor \frac{k-j}{2} \rfloor-1} \dim{\cH^b_{k-2l-j-2}} \dim{\cH^f_{j}}.
\]
Using this we can rewrite (\ref{proofdimform}) as
\[
\dim \cH_k = \dim \cH_{k-2} + \sum_{i=0}^{k} \dim{\cH^b_{k-i}} \dim{\cH^f_{i}}.
\]

We prove that this formula holds. The left-hand side equals
\[
LH =  \sum_{i=0}^k \binom{2n}{i} \binom{k-i+m-1}{m-1} -  \sum_{i=0}^{k-2} \binom{2n}{i} \binom{k-i-2+m-1}{m-1}.
\]

The right-hand side is calculated as 
\begin{eqnarray*}
RH &=& \sum_{i=0}^{k-2} \binom{2n}{i} \binom{k-i-2+m-1}{m-1} -  \sum_{i=0}^{k-4} \binom{2n}{i} \binom{k-i-4+m-1}{m-1}\\
&& + \sum_{i=0}^k \left(\binom{2n}{i} - \binom{2n}{i-2} \right) \left( \binom{k-i+m-1}{m-1} -\binom{k-i-2+m-1}{m-1} \right).
\end{eqnarray*}

The second line in this equation is expanded as
\begin{eqnarray*}
\sum_{i=0}^k \binom{2n}{i} \binom{k-i+m-1}{m-1} - \sum_{i=0}^{k-2} \binom{2n}{i} \binom{k-i-2+m-1}{m-1}\\
-\sum_{i=0}^{k-2} \binom{2n}{i} \binom{k-i-2+m-1}{m-1}+\sum_{i=0}^{k-4} \binom{2n}{i} \binom{k-i-4+m-1}{m-1}.
\end{eqnarray*}

So by adding all terms we see that the right-hand side equals the left-hand side, thus completing the proof. The case where $k >n$ is treated in a similar fashion.
\end{proof}

Now we are able to obtain the main decomposition of this section. 
First we introduce the operators
\begin{eqnarray*}
\mQ_{r,s}^k &=& \prod_{i=0, \;  i \neq k-2r-s}^{k} \dfrac{\Delta_{LB,b} + i(m-2+i)}{(i-k+2r+s)(k+i-2r-s+m-2)}\\
&&\times  \prod_{j=0, \;  j \neq s}^{\min{(n,k)}} \dfrac{\Delta_{LB,f} + j(-2n-2+j)}{(j-s)(j+s-2n-2)},
\end{eqnarray*}
with
\begin{eqnarray*}
\Delta_{LB,b} &=& \ux^2 \Delta_b - \mE_b(m-2 + \mE_b)\\
\Delta_{LB,f} &=& \uxb^2 \Delta_f - \mE_f(-2n-2 + \mE_f)
\end{eqnarray*}
the bosonic resp. fermionic Laplace-Beltrami operator.
The decomposition is then given in the following theorem.

\begin{theorem}[Decomposition of $\cH_k$]
Under the action of $SO(m) \times Sp(2n)$ the space $\cH_k$ decomposes as
\begin{equation}
\cH_{k} = \bigoplus_{i=0}^{\min(n,k)} \cH^b_{k-i} \otimes \cH^f_{i} \;\; \oplus \;\; \bigoplus_{j=0}^{\min(n, k-1)-1} \bigoplus_{l=1}^{\min(n-j,\lfloor \frac{k-j}{2} \rfloor)} f_{l,k-2l-j,j} \cH^b_{k-2l-j} \otimes \cH^f_{j},
\label{decompintoirreps}
\end{equation}
with $f_{l,k-2l-j,j}$ the polynomials determined in lemma \ref{polythm}.

Moreover, all direct summands in this decomposition are irreducible under the action of $SO(m) \times Sp(2n)$ and one has
\[
\mQ_{r,s}^k \left( f_{l,k-2l-j,j} \cH^b_{k-2l-j} \otimes \cH^f_{j} \right)= \delta_{rl} \delta_{sj} f_{l,k-2l-j,j} \cH^b_{k-2l-j} \otimes \cH^f_{j}.
\]
\label{completedecomp}
\end{theorem}

\begin{proof}
Using lemma \ref{polythm} we see that the right-hand side is contained in the left-hand side.  
Moreover we have that all summands are mutually disjoint. Indeed, as the bosonic and fermionic Laplace-Beltrami operators $\Delta_{LB,b}$ and $\Delta_{LB,f}$ both commute with $\ux^2$ and $\uxb^2$ we have that
\begin{eqnarray*}
&&\Delta_{LB,b}\Delta_{LB,f} f_{l,k-2l-j,j} \cH^b_{k-2l-j} \otimes \cH^f_{j}\\&=& f_{l,k-2l-j,j} \Delta_{LB,b}\cH^b_{k-2l-j} \otimes \Delta_{LB,f} \cH^f_{j}\\
&=& (k-2l-j)(m-2+k-2l-j)(j)(-2n-2+j)  f_{l,k-2l-j,j} \cH^b_{k-2l-j} \otimes \cH^f_{j}
\end{eqnarray*}
and hence that
\[
\mQ_{r,s}^k \left( f_{l,k-2l-j,j} \cH^b_{k-2l-j} \otimes \cH^f_{j} \right)= \delta_{rl} \delta_{sj} f_{l,k-2l-j,j} \cH^b_{k-2l-j} \otimes \cH^f_{j},
\]
proving that all summands are disjoint.

Lemma \ref{dimcheck} then shows that the left-hand side and the right-hand side of formula (\ref{decompintoirreps}) have the same dimension, so the decomposition holds.
As to the irreducibility, the pieces  $f_{l,k-2l-j,j} \cH^b_{k-2l-j} \otimes \cH^f_{j}$ clearly transform into themselves under the action of $SO(m) \times Sp(2n)$ and they are irreducible as tensor products of irreducible representations of $SO(m)$ and $Sp(2n)$.
\end{proof}

\section{The problem of integration over the supersphere}

In this and the following sections we restrict ourselves to the case $M = m-2n \not \in -2\mN$. This assumption allows us to use the Fischer decomposition (theorem \ref{fdecomp}) and the corresponding projections (theorem \ref{projectionsfischer}).

\vspace{2mm}
We start by giving a set of properties we want an integral over the supersphere to show. The supersphere is the formal object defined by the algebraic equation $x^2 = -1$ (the bosonic version of this equation is exactly the equation of the unit-sphere in $\mR^m$).

\begin{definition}
A linear functional $T: \cP \rightarrow \mR$ is called an integration over the supersphere if it satisfies the following properties for all $f(x) \in \cP$:
\begin{enumerate}
\item $T(x^2 f(x)) = - T(f(x))$
\item $T(f(g \cdot x)) = T(f(x))$, \quad $\forall g \in SO(m)\times Sp(2n)$.
\end{enumerate}
\label{defintegral}
\end{definition}

These two properties are of course very natural. The first one says that we can work modulo $x^2+1$ (the equation of the supersphere). The second property is just the invariance of integrals under the action of $SO(m) \times Sp(2n)$ (this is the generalization of rotational invariance when integrating over the classical sphere).

We will now determine the set of all functionals satisfying these two properties. More precisely, we will obtain the following theorem. 

\begin{theorem}
The space of all linear functionals $T: \cP \rightarrow \mR$ satisfying the properties 1 and 2 of definition \ref{defintegral} is a finite-dimensional vectorspace of dimension $n+1$.
\label{intdim}
\end{theorem}

In section \ref{distinctionintegrals} we will determine a way to distinguish between these different types of integration.

We first prove the following lemma, which will be crucial for the further development.

\begin{lemma}
An integral $T$ over the supersphere of a function $f \in \cH$, with $\cH$ a subspace of $\cP$ of dimension $\dim \cH > 1$, irreducible under the action of $SO(m) \times Sp(2n)$, is always zero.
\label{schur}
\end{lemma}

\begin{proof}
If $T: \cH \rightarrow \mR$ is a linear functional satisfying the requirements of definition \ref{defintegral}, then we have that $\ker T$ is invariant under the action of $SO(m) \times Sp(2n)$. As $\ker T$ is a subspace of $\cH$, we have that either $\ker T =0$ or $\ker T = \cH$, due to the irreducibility of $\cH$. As moreover $\dim (\ker T) >0$ we have that $\ker T = \cH$ and the lemma follows. (Note that this is just an application of Schur's lemma.) 
\end{proof}

Now we will determine all possible integrations over $\cP$. Recall that we have the following Fischer decomposition
\[
\cP = \bigoplus_{j=0}^{\infty} \bigoplus_{k=0}^{\infty}  x^{2j} \cH_{k}
\]
so we have for a general integral $T$ that
\[
T(\cP) = \sum_{j=0}^{\infty} \sum_{k=0}^{\infty} (-1)^j T(\cH_{k}).
\]
It hence suffices to know how integrations on $\cH_k$ look like. In theorem \ref{completedecomp} we have found that $\cH_k$ decomposes into irreducible pieces as
\begin{equation}
\cH_{k} = \bigoplus_{i=0}^{\min(n,k)} \cH^b_{k-i} \otimes \cH^f_{i} \;\; \oplus \;\; \bigoplus_{j=0}^{\min(n, k-1)-1} \bigoplus_{l=1}^{\min(n-j,\lfloor \frac{k-j}{2} \rfloor)} f_{l,k-2l-j,j} \cH^b_{k-2l-j} \otimes \cH^f_{j}.
\label{integrationsphericalharmonics}
\end{equation}

As a consequence of lemma \ref{schur}, we see that only the one-dimensional summands in the decomposition of $\cH_k$ can give rise to an integration which is not zero. From formula (\ref{integrationsphericalharmonics}) we conclude that there are exactly $n+1$ of these summands, namely
\[
f_{i,0,0} \; \cH_0^b \otimes \cH_0^f \;\; \subset \;\; \cH_{2i}, \quad i=0,\ldots,n.
\]
On each of them, the value of the integral $T$ can be freely chosen as there are no further restrictions. We denote these chosen values by $a_i \in \mR$, $i = 0, \ldots, n$.

Hence we have reduced the problem under consideration to constructing projections of elements of $\cP$ on these one-dimensional irreducible pieces. The general form of an integration over the supersphere, satisfying definition \ref{defintegral}, can then be found as follows.

Suppose we are given a polynomial $R \in \cP$, then we perform the following projections
\[
\xymatrix{R \ar[rrr]^{\mP_{2k}} &&& \cP_{2k} \ar[rrr]^{ x^{2k-2i}\mP_{2k-2i}^{2k}}&&& x^{2k-2i} \cH_{2i} \ar[rrr]^{\mP_{2i}^{bf}}&&&f_{i,0,0} \cH_0^b \otimes \cH_0^f}
\]
with
\begin{itemize}
\item $\mP_{2k}$ the projection onto the space of homogeneous polynomials of degree $2 k$
\item $\mP_{2k-2i}^{2k}$ the projection from the space of homogeneous polynomials of degree $2 k$ to the space of spherical harmonics of degree $2i$ (Fischer decomposition)
\item $\mP_{2i}^{bf}$ the projection from the space of spherical harmonics of degree $2i$ to its unique one-dimensional irreducible subspace $f_{i,0,0} \cH_0^b\otimes \cH_0^f$.
\end{itemize}

These projections have to be done for all values of $k$ and for $i = 0, \ldots , n$. Moreover, we only have to consider projections on these pieces, as all the other components of $R$ are elements of irreducible subspaces of dimension larger than one and hence, by lemma \ref{schur}, do not contribute to the integral.

Summarizing, we arrive at the following general form for an integral on the supersphere:
\begin{equation}
T = \sum_{i=0}^{n} \frac{a_i}{f_{i,0,0}} \sum_{k=i}^{\infty} (-1)^{k-i} \mP_{2i}^{bf} \mP_{2k-2i}^{2k} \mP_{2k}.
\label{generalintegral}
\end{equation}

The factor $(-1)^{k-i}$ stems from the fact that $x^{2k-2i} \cH_{2i}$ equals $(-1)^{k-i}\cH_{2i}$ on the supersphere $x^2=-1$.

As there are exactly $n+1$ one-dimensional irreducible subspaces $f_{i,0,0} \; \cH_0^b \otimes \cH_0^f$, and thus $n+1$ values $a_i$ to be chosen, this also proves theorem \ref{intdim}.

Let us now find explicit formulae for the projection operators. The operators $\mP_{2k-2i}^{2k}$ follow from the Fischer decomposition (see theorem \ref{projectionsfischer}). Next we construct the operator $\mP_{2i}^{bf}$. This operator is the projection
\[
\mP_{2i}^{bf}: \cH_{2i} \longrightarrow f_{i,0,0} \cH_0^b \otimes \cH_0^f.
\]

It is immediately clear that $\Delta_b^i$ annihilates all terms in the decomposition of $\cH_{2i}$ except for the term $f_{i,0,0} \cH_0^b \otimes \cH_0^f$. Indeed, we have ($i >0$)
\begin{eqnarray*}
\Delta_b^i (f_{i,0,0} \cH_0^b \otimes \cH_0^f) &=& \Delta_b^i (f_{i,0,0})\cH_0^b \otimes \cH_0^f\\
&=&  \frac{n!}{\Gamma(\frac{m}{2} + i)} \Delta_b^i(\ux^{2i})\cH_0^b \otimes \cH_0^f\\
&=& \frac{n!}{\Gamma(\frac{m}{2} + i)} 2^{2i} i! \frac{\Gamma(\frac{m}{2} + i)}{\Gamma(\frac{m}{2})}\cH_0^b \otimes \cH_0^f\\
&=&\frac{n!i!2^{2i}}{\Gamma(\frac{m}{2})}\cH_0^b \otimes \cH_0^f.
\end{eqnarray*}

So the operator $\mP_{2i}^{bf}$ takes the following form
\[
\mP_{2i}^{bf} = \frac{\Gamma(\frac{m}{2})}{n!i!2^{2i}} f_{i,0,0} \Delta_b^i.
\]

Note that in the case where $i=0$ no projection is necessary ($\mP_{0}^{bf} = 1$) as $\cH_0$ is already one-dimensional.

A general integral $T$ over the supersphere can now also be written as
\[
T = \sum_{i=0}^n a_i \int_i
\]
with 
\[
\int_i = \sum_{k=i}^{\infty} \frac{(-1)^{k-i}}{f_{i,0,0}} \mP_{2i}^{bf} \mP_{2k-2i}^{2k} \mP_{2k}
\]
such that 
\[
\int_i f_{j,0,0} = \delta_{ij}.
\]

In the next section we will construct some explicit examples.

\section{Some examples}

\subsection{The Pizzetti case}

This integral is defined by putting $a_0 \neq 0$ and $a_i =0$, $i>0$. When explicitly writing the proper projections we arrive at the following

\begin{definition}[Pizzetti]
The integral of $R\in \cP$ over the supersphere is given by
\begin{equation}
\int_{SS} R =  \sum_{k=0}^{\infty} (-1)^k \frac{2 \pi^{M/2}}{4^{k} k!\Gamma(k+M/2)} (\Delta^{k} R)(0)
\label{defintss}
\end{equation}
where $(\Delta^{k} R)(0)$ means evaluating the result in $x_i = {x \grave{}}_i =0$.
\end{definition}

The normalization is chosen such that $\int_{SS}  1$ gives the area of the sphere in the purely bosonic case.
The same formula was proven in the classical case by Pizzetti (see \cite{PIZZETTI}).

Combining this formula with the concept of integration in spherical co-ordinates yields the Berezin integral (see \cite{DBS5} and \cite{MR732126}), which is the standard integral used in the study of superspaces. More precisely we have the following theorem.

\begin{theorem}
The integral of a function $f = R \, \exp(x^2)$ with $R$ an arbitrary super-polynomial is given by the following formula:
\begin{equation}
\int_{\mR^{m|2n}} f = \sum_{k=0}^{\infty} (-1)^k \frac{\pi^{M/2}}{4^{k} k!} (\Delta^{k} R )(0) =\pi^{M/2}(\exp(-\Delta/4))R(0).
\label{superintegral}
\end{equation}
Moreover this integral is equivalent with the Berezin integral:
\[
\int_{\mR^{m|2n}} f = \pi^{-n}  \int_{\mR^m} dV(\ux) \partial_{{x \grave{}}_{2n}} \ldots \partial_{{x \grave{}}_{1}} f
\]
where $dV(\ux)$ is the standard Lebesgue measure in $\mR^m$.
\end{theorem}

\begin{proof}
See \cite{DBS5}.
\end{proof}

Note that in the resulting formula (\ref{superintegral}) the super-dimension $M$ does not appear, except in the scaling of the result.

In the following subsection we establish another functional leading to a possible new concept of integration in superspace.

\subsection{Another possibility}

We consider the simplest case (with exception of the Pizzetti integral).
This is the integral where $a_1 \neq 0$ but $a_i =0, i \neq 1$. Explicitly determining the projections in formula (\ref{generalintegral}) leads to

\begin{definition}
The integral of a superpolynomial $R$ over the supersphere is given by
\begin{equation}
\int_{1} R = c \sum_{k=1}^{\infty}  \frac{(-1)^{k-1}}{4^{k+1} (k-1)!\Gamma(k+M/2+1)} \left(\Delta_b( 2M \Delta^{k-1} - x^2 \Delta^k) R \right)(0).
\label{defintesch}
\end{equation}
\end{definition}

This can be further simplified to
\[
\int_{1} R = c \sum_{k=1}^{\infty}  \frac{(-1)^{k-1}}{4^{k+1} (k-1)!\Gamma(k+M/2+1)} \left( 2M \Delta_b\Delta^{k-1} - 2m \Delta^k) R \right)(0).
\]

In this definition $c$ is a constant which still has to be determined. If we take 
\[
c = \frac{\Gamma(2 + M/2)\Gamma(1 + m/2)}{m M n!}
\]
then the integral is normalized such that
\[
\int_1 f_{1,0,0} = 1.
\]

Again combining this idea with integration in spherical co-ordinates, we find the following formula for integration over the whole superspace (which is now not equivalent with the Berezin integral)
\[
\int_{1, \mR^{m | 2n}} R \exp(x^2)= \frac{c}{2} \sum_{k=1}^{\infty}  \frac{(-1)^{k-1}}{4^{k+1} (k-1)!(k+M/2)} \left(\Delta_b( 2M \Delta^{k-1} - x^2 \Delta^k) R \right)(0).
\]

In this case the super-dimension $M$ does not disappear in the resulting formula.

In a similar way one can construct the integrals corresponding to $a_j \neq 0$, $a_i =0 \, (i \neq j)$ explicitly.

\section{Distinction between the different types of integration}
\label{distinctionintegrals}

In the purely bosonic case ($n=0$), there is only one possibility for integration on the sphere, namely Pizzetti's formula, because there is only one one-dimensional space $\cH_0^b$. In this section we will show how this particular possibility can be distinguished in general from the other ones.

We start with the following definition of orthogonality:

\begin{definition}
The space $\cH_k$ is orthogonal to $\cH_l$ ($k \neq l$), notation $\cH_k \; \bot \; \cH_l$, with respect to the integral $T$ over the supersphere if and only if
\[
T(\cH_k \cH_l) = 0 = T(\cH_l \cH_k).
\]
\end{definition}

We then have the following

\begin{theorem}
The Pizzetti integral over the supersphere is the only integral that has the property
\begin{equation}
k \neq l \qquad \Longrightarrow \qquad\cH_k \; \bot \; \cH_l.
\label{orthsphharm}
\end{equation}
\end{theorem}

\begin{proof}
The fact that the Pizzetti integral satisfies (\ref{orthsphharm}) is proven in \cite{DBS5}, theorem 4. 

Conversely, a general integral on the supersphere has the following form:
\[
\int = \sum_{i=0}^{n} \frac{a_i}{f_{i,0,0}} \sum_{k=i}^{\infty} (-1)^{k-i} \mP_{2i}^{bf} \mP_{2k-2i}^{2k} \mP_{2k}.
\]
If it is not the Pizzetti integral, then there exists a maximal $t >0$ such that $a_t \neq 0$. Then one has that
\[
\int f_{t,0,0} = a_t \neq 0
\]
which immediately implies that $\cH_0$ is not orthogonal to $\cH_{2t}$.
\end{proof}

One can even go a step further. Not only the spaces $\cH_k$ are mutually orthogonal with respect to the Pizzetti integral, but in fact all irreducible pieces (see theorem \ref{completedecomp}) are mutually orthogonal. This is summarized in the following theorem. The proof is essentially a reduction to either the purely bosonic or the purely fermionic case, treated in \cite{DBS5}.

\begin{theorem}
If $(i,p,q) \neq (j,r,s)$, one has that 
\[
f_{i,p,q} \cH^b_{p} \otimes \cH^f_{q} \quad  \bot \quad f_{j,r,s} \cH^b_{r} \otimes \cH^f_{s}
\]
with respect to the Pizzetti integral.
\end{theorem}

\begin{proof}
It is only necessary to prove this for the irreducible pieces contained in the same $\cH_k$. So we prove that
\[
f_{i,k-2i-p,p} \cH^b_{k-2i-p} \otimes \cH^f_{p} \quad  \bot \quad f_{j,k-2j-q,q} \cH^b_{k-2j-q} \otimes \cH^f_{q}
\]
with either $p \neq q$ or $p=q$, $i \neq j$. Due to the definition of the Pizzetti integral (formula (\ref{defintss})) it suffices to prove that
\[
\Delta^k (f_{i,k-2i-p,p} \cH^b_{k-2i-p} \otimes \cH^f_{p} \; f_{j,k-2j-q,q} \cH^b_{k-2j-q} \otimes \cH^f_{q}) =0.
\]

As we also have that
\begin{eqnarray*}
f_{i,k-2i-p,p} &=& \sum_{s=0}^i a_s \ux^{2s} \uxb^{2i-2s} \\
f_{j,k-2j-q,q} &=& \sum_{t=0}^j b_t \ux^{2t} \uxb^{2j-2t}
\end{eqnarray*}
it is sufficient to consider a term of the form
\begin{eqnarray*}
\Delta^k \left( \ux^{2s}\uxb^{2i-2s} \cH^b_{k-2i-p} \otimes \cH^f_{p} \ux^{2t} \uxb^{2j-2t} \cH^b_{k-2j-q} \otimes \cH^f_{q} \right)\\
=\Delta^k \left( \ux^{2s+2t} \cH^b_{k-2i-p} \cH^b_{k-2j-q}   \uxb^{2i+2j-2s-2t}  \cH^f_{p} \cH^f_{q} \right).
\end{eqnarray*}

Now $\Delta^k$ can be expanded as
\[
\Delta^k = \sum_{u=0}^{k} c_u \Delta_b^{k-u} \Delta_f^u.
\]

If $p+q$ is odd, then all terms vanish; if $p+q$ is even there remains exactly one term, namely where $2u = 2i+2j-2s-2t+p+q$. We obtain
\begin{eqnarray*}
\Delta_b^{k-u} \left( \ux^{2s+2t} \cH^b_{k-2i-p} \cH^b_{k-2j-q} \right) \Delta_f^{u} \left(  \uxb^{2i+2j-2s-2t}  \cH^f_{p} \cH^f_{q} \right)\\
= \mbox{constant} \times \Delta_b^{k-i-j-\frac{p+q}{2}} \left( \cH^b_{k-2i-p} \cH^b_{k-2j-q} \right) \Delta_f^{\frac{p+q}{2}} \left(    \cH^f_{p} \cH^f_{q} \right),
\end{eqnarray*}
where we have used formula (\ref{relationslaplace}) in the second line.

If $p \neq q$, the second term is always zero (apply \cite{DBS5}, theorem 4 in the purely fermionic case); if $p=q$ then $i \neq j$ and the first term is always zero (apply \cite{DBS5}, theorem 4 in the purely bosonic case). This shows that both spaces are indeed orthogonal.
\end{proof}

We can now summarize the previous results in the following theorem:

\begin{theorem}
If $M \not \in -2\mN$, the only linear functional $T: \cP \rightarrow \mR$ satisfying the following properties for all $f(x) \in \cP$:
\begin{itemize}
\item $T(x^2 f(x)) = - T(f(x))$
\item $T(f(g \cdot x)) = T(f(x))$, \quad $\forall g \in SO(m)\times Sp(2n)$
\item $k \neq l \quad \Longrightarrow \quad T(\cH_k \cH_l) = 0 = T(\cH_l \cH_k)$
\item $T(1) = \dfrac{2 \pi^{M/2}}{\Gamma(M/2)}$,
\end{itemize}
is given by the Pizzetti integral
\[
\int_{SS}R =  \sum_{k=0}^{\infty} (-1)^k \frac{2 \pi^{M/2}}{4^{k} k!\Gamma(k+M/2)} (\Delta^{k} R)(0).
\]
\end{theorem}

\section*{Acknowledgement}
Part of this research was done during a research visit of the first two authors to Charles University, Prague. The first author acknowledges support by the institutional grant BOF06/BIL/012, the second author was supported by a Travel Grant of the Research Foundation - Flanders (FWO).


\begin{thebibliography}{10}

\bibitem{MR0208930}
{\sc Berezin F A} 1966
\newblock {\em The method of second quantization}.
\newblock Pure and Applied Physics, Vol. 24. Academic Press, New York

\bibitem{MR732126}
{\sc Berezin F A} 1983
\newblock {\em Introduction to algebra and analysis with anticommuting
  variables}.
\newblock Moskov. Gos. Univ., Moscow


\bibitem{MR0580292}
{\sc Kostant B} 1977
\newblock Graded manifolds, graded {L}ie theory, and prequantization.
\newblock In {\em Differential geometrical methods in mathematical physics
  (Proc. Sympos., Univ. Bonn, Bonn, 1975)}. Springer, Berlin,
  pp.~177--306. Lecture Notes in Math., Vol. 570.

\bibitem{MR565567}
{\sc Le{\u\i}tes D A} 1980
\newblock Introduction to the theory of supermanifolds.
\newblock {\em Uspekhi Mat. Nauk 35}, 1(211), 3--57, 255.

\bibitem{MR778559}
{\sc DeWitt B} 1984
\newblock {\em Supermanifolds}.
\newblock Cambridge Monographs on Mathematical Physics. Cambridge University
  Press, Cambridge.

\bibitem{MR574696}
{\sc Rogers A} 1980
\newblock A global theory of supermanifolds.
\newblock {\em J. Math. Phys. 21}, 1352--1365.

\bibitem{MR697564}
{\sc Brackx F, Delanghe R and Sommen F} 1982
\newblock {\em Clifford analysis}, vol.~76 of {\em Research Notes in
  Mathematics}.
\newblock Pitman (Advanced Publishing Program), Boston, MA.

\bibitem{MR1169463}
{\sc Delanghe R, Sommen F and Sou{\v{c}}ek V} 1992
\newblock {\em Clifford algebra and spinor-valued functions}, vol.~53 of {\em
  Mathematics and its Applications}.
\newblock Kluwer Academic Publishers Group, Dordrecht.

\bibitem{BSSVL} 
{\sc Bure\v{s} J, Sommen F, Sou\v{c}ek V and Van Lancker P} 2001 
\newblock Symmetric analogues of Rarita-Schwinger equations. 
\newblock {\em Ann. Global Anal. Geom. 21\/}, 215-240.

\bibitem{DBS1}
{\sc  De~Bie H and Sommen F} 2007
\newblock Correct rules for Clifford calculus on superspace.
\newblock {\em Adv. Appl. Clifford Algebr. 17\/}, 357-382.

\bibitem{DBS4}
{\sc  De~Bie H and Sommen F} 2007
\newblock A Clifford analysis approach to superspace.
\newblock {\em Ann. Physics. 322\/}, 2978-2993.

\bibitem{DBS2}
{\sc  De~Bie H and Sommen F} 2008
\newblock Fischer decompositions in superspace.
\newblock In {\em Function spaces in complex
and Clifford analysis}, National University Publishers Hanoi, pp. 170-188.

\bibitem{DBS5}
{\sc  De~Bie H and Sommen F} 2007
\newblock Spherical harmonics and integration in superspace.
\newblock {\em J. Phys. A: Math. Theor. 40\/}, 7193-7212.

\bibitem{PIZZETTI}
{\sc Pizzetti P} 1909
\newblock Sulla media dei valori che una funzione dei punti dello spazio assume
  alla superficie di una sfera.
\newblock {\em Rend. Lincei 18\/}, 182--185.

\bibitem{MR784620}
{\sc Rabin J M} 1985
\newblock The {B}erezin integral as a contour integral.
\newblock {\em Phys. D 15}, 65--70.

\bibitem{MR825156}
{\sc Rogers A} 1986
\newblock Realizing the {B}erezin integral as a superspace contour integral.
\newblock {\em J. Math. Phys. 27}, 710--717.

\bibitem{DBS3}
{\sc  De~Bie H and Sommen F} 2007
\newblock Hermite and Gegenbauer polynomials in superspace using Clifford
  analysis.
\newblock {\em J. Phys. A: Math. Theor. 40\/}, 10441-10456.

\bibitem{MR1151617}
{\sc Howe R and Tan E C} 1992
\newblock {\em Nonabelian harmonic analysis}.
\newblock Universitext. Springer-Verlag, New York.

\bibitem{MR0229863}
{\sc Vilenkin N J} 1968
\newblock {\em Special functions and the theory of group representations}.
\newblock Translations of Mathematical Monographs, Vol. 22. American Mathematical Society, Providence, R. I..

\bibitem{MR1153249}
{\sc Fulton W and Harris J} 1991
\newblock {\em Representation theory}, vol.~129 of {\em Graduate Texts in
  Mathematics}.
\newblock Springer-Verlag, New York.

\bibitem{MR1827871}
{\sc Dunkl C F and Xu Y} 2001
\newblock {\em Orthogonal polynomials of several variables}, vol.~81 of 
{ \em Encyclopedia of Mathematics and its Applications}.
\newblock Cambridge University Press, Cambridge.

\bibitem{MR2207700}
{\sc Ben~Sa{\"{\i}}d S and {\O}rsted B} 2006
\newblock Segal-{B}argmann transforms associated with finite {C}oxeter groups.
\newblock {\em Math. Ann. 334}, 281--323.

\end{thebibliography}
\end{document}